\documentclass[11pt]{article}
\newcommand{\PAPER}[1]{#1}
\newcommand{\SIAMV}[1]{}

\PAPER{

\usepackage{fullpage,times,amsthm,amssymb,hyperref}
\usepackage[T1]{fontenc}

\title{Derandomizing Pseudopolynomial Algorithms for Subset Sum}

\author{Timothy M. Chan\thanks{
Siebel School of Computing and Data Science, University of Illinois at Urbana--Champaign, USA (tmc@illinois.edu)
}}

\newtheorem{lemma}{Lemma}[section]
\newtheorem{theorem}[lemma]{Theorem}
\newtheorem{corollary}[lemma]{Corollary}
\newtheorem{definition}[lemma]{Definition}
\newtheorem{observation}[lemma]{Observation}

\newcommand{\mysection}[1]{\section{#1}}
\newcommand{\mysubsection}[1]{\subsection{#1}}
\newcommand{\myparagraph}[1]{\paragraph{#1}}

}

\newcommand{\SSS}{\mathcal{C}}
\newcommand{\hatS}{\hat{S}}
\newcommand{\OO}{\widetilde{O}}
\newcommand{\eps}{\varepsilon}
\newcommand{\SUM}{\sigma}
\newcommand{\SUMS}{\Sigma}
\newcommand{\ANS}{\mbox{\sc out}}

\begin{document}

\maketitle

\begin{abstract}
We reexamine the classical \emph{subset sum} problem: given a set $X$ of $n$ positive integers and a number~$t$, decide whether there exists a subset of $X$ that sums to $t$; or more generally, compute the set $\textsc{out}$ of all numbers $y\in\{0,\ldots,t\}$ for which there exists a subset of $X$ that sums to $y$. Standard dynamic programming solves the problem in $O(tn)$ time. In SODA'17, two papers appeared giving the current best deterministic and randomized algorithms, ignoring polylogarithmic factors: Koiliaris and Xu's deterministic algorithm runs in $\widetilde{O}(t\sqrt{n})$ time, while Bringmann's randomized algorithm runs in $\widetilde{O}(t)$ time. We present the first deterministic algorithm running in $\widetilde{O}(t)$ time.

Our technique has a number of other applications: for example, we can also derandomize the more recent output-sensitive algorithms by Bringmann and Nakos [STOC'20] and Bringmann, Fischer, and Nakos [SODA'25] running in $\OO(|\textsc{out}|^{4/3})$ and $\OO(|\textsc{out}|\sqrt{n})$ time, and we can derandomize a previous fine-grained reduction from 0-1 knapsack to min-plus convolution by Cygan et al.~[ICALP'17].
\end{abstract}

\mysection{Introduction}

This paper studies one of the textbook problems in algorithm design, the \emph{subset sum} problem:
given a (multi)set $X$ of $n$ positive integers and a target number $t$, decide whether there
exists a subset of $X$ that sums to~$t$.  The problem is weakly NP-complete.  
Standard algorithms achieve roughly $O(2^{n/2})$ time (by the meet-in-the-middle
trick~\cite{HorowitzS74}) and $O(tn)$ time (by dynamic programming~\cite{BellmanBOOK}).  Although it has remained open whether the exponential bound
$2^{n/2}$ could be improved in the worst case, the pseudopolynomial bound $O(tn)$ has been improved: in two
concurrent SODA'17 papers,
Koiliaris and Xu~\cite{KoiliarisX19} described a deterministic $\OO(t\sqrt{n})$-time algorithm,\footnote{
Throughout this paper, the $\OO$ notation hides factors polylogarithmic in $n$ and $t$.
} and
Bringmann~\cite{Bringmann17} described a randomized (Monte Carlo) $\OO(t)$-time algorithm.
Another elegant randomized (Monte Carlo) $\OO(t)$-time algorithm was found by Jin and Wu~\cite{JinW19}.

All of these pseudopolynomial algorithms can solve a stronger ``all-targets'' version of the problem:
compute $\ANS = \{y\in\{0,\ldots,t\}: \mbox{there exists a subset of $X$ that sums to $y$}\}$.
They can also compute a representation of a feasible subset for each $y\in\ANS$
(so that a subset may subsequently be reported in $\OO(1)$ time per element).
For the all-targets version of the problem, the randomized $\OO(t)$-time algorithms are clearly near optimal since the output size is $\Omega(t)$ in the worst case.  For the original version of the problem,
Abboud, Bringmann, Hermelin, and Shabtay~\cite{AbboudBHS22} also ruled out algorithms with time bounds
of the form $O(t^{1-\eps} 2^{o(n)})$ for any constant $\eps>0$, assuming SETH (see also \cite{CyganDLMNOPSW16} for an earlier 
conditional lower bound under a different hypothesis).

In recent years, numerous other algorithms have been proposed for
the subset sum problem that take into account different parameters (e.g., the output size $|\ANS|$
\cite{BringmannN20,BringmannFN25}, the largest input value $\max(X)$ \cite{PolakRW21,ChenLMZ24a,ChenLMZ24},
the density~\cite{BringmannW21}, space usage \cite{LokshtanovN10,Bringmann17,JinVW21,NederlofW21,BelovaCKM24}, \ldots).
Many variants have also been explored (e.g., unbounded subset sum~\cite{Bringmann17}, modular subset sum~\cite{AxiotisBJTW19,AxiotisBBJNTW21,CardinalI21}, approximate subset sum~\cite{KellererMPS03,ChenLMZ24a}, \ldots),
and the problem is connected to other fundamental problems such as knapsack and coin changing, forming
a very active topic of current research.

\myparagraph{Derandomization?}
The results in the two SODA'17 papers naturally raise 
the question of whether there exists a \emph{deterministic} $\OO(t)$-time algorithm for subset sum, and 
this question has remained open for the last 8 years.
Bringmann's randomized algorithm uses a color-coding technique~\cite{AlonYZ95}, and although
some past algorithms based on color-coding have been successfully derandomized (e.g., for computing $k$-simple
paths/cycles in graphs), the overhead for derandomizing color-coding seems too big for the application to
subset sum (our desired $\OO(t)$ bound doesn't leave much room to absorb the overhead).  If the number of elements $k$ allowed in a sum
is sublogarithmic, the cost of derandomization is tolerable, but not when $k$ is large.
It is unclear how to
apply standard techniques such as the method of limited independence or the method of conditional probabilities/expectations without blowing up the time bound.
Jin and Wu's alternative randomized algorithm~\cite{JinW19} seems even tougher to derandomize---it computes counts modulo a random prime, and there is no obvious efficient way to guarantee that a correct prime is selected.
Just converting these Monte Carlo randomized algorithms to Las Vegas is already open.
On the other hand, Koiliaris and Xu's deterministic approach~\cite{KoiliarisX19} (see also \cite{KoiliarisX18} for an alternative) seems difficult to improve upon---they did give a variant with $\OO(t^{4/3})$ running time but there is no clear path towards $O(t^{1+\eps})$.  By contrast, for variants such as unbounded subset sum\footnote{In the unbounded subset sum problem, we may use an element multiple times to sum to the target.  This tends to be easier, as much of the challenge in the original problem is in preventing an element from being used more than once.
}~\cite{Bringmann17}
and modular subset sum~\cite{AxiotisBBJNTW21}, efficient derandomization is not as difficult and is already known.

\myparagraph{New algorithm.}
We present the first deterministic algorithm for subset sum running in $\OO(t)$ time.  It also solves the all-targets version of the problem.

Our algorithm may be viewed as a derandomization of a variant of Bringmann's~\cite{Bringmann17}.  Instead of partitioning the input set into multiple disjoint subsets via color-coding, we recursivly partition the input set into two subsets, in such a way that the desired subset for each target is split roughly evenly into two halves in terms of their cardinalities (with about $\sqrt{k}$ additive error if the subset has about $k$ elements).  This reinterpretation of Bringman's algorithm as a divide-and-conquer via ``halvers'' (or similar randomized partitioning) is not new, and has appeared in other subsequent randomized algorithms for subset sum and knapsack (e.g., see \cite{BringmannN20,BringmannDP24,BringmannFN25}).  However, although the existence of a good halver
is guaranteed by the probabilistic method (via standard Chernoff or discrepancy bound), its deterministic construction seems to require knowing the subsets for all targets, i.e., knowing all the answers in the first place!
Our new approach
 is to speed up the halver construction by worsening its quality.\footnote{
The idea is very loosely inspired by work on derandomization in computational geometry concerning
``$\eps$-nets'' and ``$\eps$-approximations'' \cite{ChazelleBOOK,MatousekSURV}.
One specific inspiration is the author's derandomization of a recursive algorithm for linear programming in low dimensions~\cite{Chan18}, which builds a lower-quality $\eps$-net (sufficient for this recursive algorithm) simply by dividing the input set into small subsets and applying a slow $\eps$-net algorithm to each subset.
}  So long as the additive error is kept slightly sublinear (bounded by $k/\log k$), the divide-and-conquer would still be efficient.
To this end, we show that for any subset, there is an equivalent subset (with the same sum) that
can be divided into a sublinear number of small ``canonical'' subsets; we then naively construct a halver with respect to these canonical subsets. The quality of the halver with respect to the original subsets would indeed suffer, but since there are not too many canonical subsets overall and their sizes are small, the construction is cheaper.  But how to generate good canonical subsets?  Our idea is to adapt Koiliaris and Xu's deterministic algorithm for this purpose!  Putting these pieces together in the right way, the whole solution turns out to be relatively simple, 
describable in 3 pages (see Section~\ref{sec:linear}).

\myparagraph{Other consequences.}
Our new idea on halver construction can help derandomize other algorithms, too.  We mention two applications:
\begin{enumerate}
\item Bringmann and Nakos \cite{BringmannN20} and Bringmann, Fischer, and Nakos \cite{BringmannFN25}
recently obtained randomized \emph{output-sensitive} algorithms for the all-targets version of the subset sum problem,
running in $\OO(|\ANS|^{4/3})$ and
$\OO(|\ANS|\sqrt{n})$ time respectively.
(Standard dynamic programming~\cite{BellmanBOOK} can be made to run in $\OO(|\ANS|n)$ time.)
Our technique can derandomize both algorithms with the same time bound up to polylogarithmic factors in $n$ and $t$.
\item Cygan, Mucha, W\k{e}grzycki, and W\l{}odarczyk~\cite{CyganMWW19} gave a randomized reduction showing that if
the min-plus convolution problem for $n$ elements could be solved in $O(n^{2-\eps})$ time, then
the 0-1 knapsack problem for items with positive integer weights bounded by $t$ and positive real profits could be solved
in randomized $\OO(t^{2-\eps'})$ time.  Our technique provides a derandomization of this result.
(Derandomizing reductions have gained more attention recently; e.g., see \cite{ChanH20,ChanX24,FischerKP24}.)
In particular, this implies a deterministic $O(t^2/2^{\Omega(\sqrt{\log t})})$-time algorithm for
0-1 knapsack, by using Chan and Williams' deterministic min-plus matrix multiplication algorithm~\cite{ChanW21}
together with a known reduction from min-plus matrix multiplication to convolution~\cite{BremnerCDEHILPT14}.
\end{enumerate}
The second application above follows immediately by our technique (see Appendix~\ref{sec:knapsack}).  The first application is more interesting, however, and requires additional 
new ideas, as we can no longer apply Koiliaris and Xu's approach to generate canonical subsets (see Section~\ref{sec:os}).

\myparagraph{Notation.}
Let $[u]$ denote $\{0,1,\ldots,u\}$.
For a set $S$ of integers, let $\SUM(S)=\sum_{s\in S} s$.
Define 
\[ \SUMS_i(X) \:=\: \{y: \mbox{$\exists S\subseteq X$
with $|S|=i$ and $\SUM(S)=y$}\}.
\]
Let $\SUMS_{\le k}(X) := \bigcup_{i\in [k]}\SUMS_i(X)$, and $\SUMS(X)=\SUMS_{\le\infty}(X)$.
The all-targets version of the subset sum problem is to compute $\SUMS(X)\cap [t]$ for a given $t$.

\mysection{An $\OO(t)$-Time Algorithm}\label{sec:linear}

In this section, we present a deterministic $\OO(t)$-time algorithm for all-targets subset sum.

\mysubsection{Preliminaries: Review of Koiliaris and Xu's algorithm}

Many previous work on subset sum starts by considering the related problem of computing $\SUMS_{\le k}(X)$ (or a superset) for a given $k$.  Bringmann~\cite{Bringmann17} solved this version of the problem in randomized $\OO(ku)$ time, where $u=\max(X)$, whereas Koiliaris and Xu~\cite{KoiliarisX19} solved it in deterministic $\OO(k^2u)$ time.
We first review Koiliaris and Xu's solution, as we will need it later.
Their approach is a binary divide-and-conquer, using convolutions as a subroutine.  We describe a simplified version below,
running in $\OO(k^3u)$ time, which is good enough for our application; we reinterpret the divide-and-conquer 
in terms of dyadic intervals.


\begin{lemma}\label{lem:KX}
Given a set $X$ of integers in $[u]$ and a number $k$,
we can compute $\SUMS_i(X)$ 
for all $i\in [k]$, in $\OO(k^3 u)$ deterministic time. 

Afterwards, for any given $i\in[k]$ and $y\in \SUMS_i(X)$, we can report one
subset $S\subseteq X$ with $|S|=i$ and $\SUM(S)=y$, in $\OO(k)$ time.
\end{lemma}
\begin{proof}
W.l.o.g., assume that $u$ is a power of 2.  A \emph{dyadic interval} is an interval
of the form $(a,a+\ell]$ where the length $\ell$ is a power of 2 and $a$ is divisible by $\ell$.
For each dyadic interval $I=(a,a+\ell]\subseteq (0,u]$,
we want to compute $\{\SUMS_i(X\cap I)\}_{i\in[k]}$.

Decompose $I$ into two dyadic subintervals $I'=(a,a+\ell/2]$ and $I''=(a+\ell/2,a+\ell]$.  
We can compute $\{C_i := \SUMS_i(X\cap I)\}_{i\in[k]}$ from $\{A_{i'}:=\SUMS_{i'}(X\cap I')\}_{i'\in [k]}$
and $\{B_{i''}:=\SUMS_{i''}(X\cap I'')\}_{i''\in[k]}$ by the formula
$C_i = \bigcup_{i',i'':\: i'+i''=i} (A_{i'}+B_{i''})$.
Since $A_{i'}\subseteq (ai',(a+\ell/2)i']$ and $B_{i''}\subseteq ((a+\ell/2)i'',(a+\ell)i'']$ are contained in intervals of length $O(k\ell)$ for each $i',i''\in [k]$,
we can compute each sumset $A_{i'}+B_{i''}$ via a convolution of two Boolean vectors of length $O(k\ell)$ in $O(k\ell\log(k\ell))$ time.
For each element $z\in A_{i'}+B_{i''}$, we can also compute a \emph{witness} $(x,y)\in A_{i'}\times B_{i''}$ with $x+y=z$, 
in $O(k\ell\log^{O(1)}(k\ell))$ time using the deterministic witness-finding technique by Alon and Naor~\cite{AlonN96}.\footnote{Alon and Naor's witness-finding technique was originally described for matrix multiplication, but the technique is general and applies to convolution.}
The total cost over all $O(k^2)$ choices of $(i',i'')$ is $O(k^3\ell\log^{O(1)}\ell)$.
For each element $z\in C_i$, we can also record an $(i',i'')$ with $i'+i''=i$ and $z\in A_{i'}+B_{i''}$.

Proceeding bottom-up over all dyadic intervals, which have total length $O(u\log u)$,
we obtain a total running time of $\OO(k^3 u)$.
For the root interval $I_0=(0,u]$,
we can report a subset corresponding to a given $i\in [k]$ and $y\in \SUMS_i(X\cap I_0)$ in $\OO(k)$ time by backtracking using 
the witnesses to the convolutions.
\end{proof}

\PAPER{\noindent} 
(\emph{Note}.  Koiliaris and Xu's original version of the algorithm~\cite{KoiliarisX19} achieved $\OO(k^2u)$ time instead of $\OO(k^3u)$, essentially
by observing that the $C_i$'s can be computed by a single convolution of vectors of length $O(k^2\ell)$, instead of $O(k^2)$ convolutions of vectors of length $O(k\ell)$.  Also, they performed the divide-and-conquer slightly differently,
by sorting the elements $x_1,\ldots,x_n$ of $X$ and computing $\{\SUMS_i(\{x_j:j\in I\})\}_{i\in[k]}$ for dyadic
intervals $I\subseteq (0,n]$.  This allows some $\log u$ factors to be replaced by $\log n$, which isn't important to us,
although this alternative will be useful later in Section~\ref{sec:os}.)

\mysubsection{Halver}

The key idea behind our algorithm concerns the construction of a \emph{halver}, which we define below:

\begin{definition}\label{def:halver}
Given a set $X$ of positive integers, and numbers $k$ and $\Delta$,
a \emph{$(k,\Delta)$-halver} of $X$ is a partition of $X$ into two disjoint subsets $X'$ and $X''$ with the following property:
\begin{quote}
For every subset $S\subseteq X$ with $|S|\le k$, 
there exists a subset $\hatS\subseteq X$ with $|\hatS|\le k$, $\SUM(\hatS)=\SUM(S)$,
and $|\hatS\cap X'|, |\hatS\cap X''|\le |\hatS|/2 + \Delta$.
\end{quote}
\end{definition}

Note the slight subtlety with the existential quantifier in the above definition (since there could be multiple
subsets with the same cardinality and the same sum as $S$).  
We recall a well-known fact from discrepancy theory:

\begin{lemma}\label{lem:discrep}
Given a set $X$ and a collection $\SSS$ of subsets of $X$ each of size at most $k$,
we can compute a partition of $X$ into two subsets $X'$ and $X''$ such that
$|S\cap X'|, |S\cap X''|\le |S|/2 + O(\sqrt{|S|\log|\SSS|})$ for all $S\in\SSS$, in $O(k|\SSS| + |X|)$ deterministic time.
\end{lemma}
\begin{proof}
The existence of such a partition $(X',X'')$ is a basic fact in discrepancy theory~\cite{ChazelleBOOK}.  (A random partition works with high probability by a simple application of the Chernoff bound.)  A deterministic construction of $(X',X'')$ can be obtained by the method of conditional expectation; see~\cite[Theorem 1.2]{ChazelleBOOK} for a good description.%
\footnote{The time bound stated there was $O(|X||\SSS|)$, but because all our subsets have sizes at most $k$, the running time of the algorithm is actually $O(k|\SSS|+|X|)$.}
\end{proof}

The above lemma implies the existence of a $(k,\OO(\sqrt{k}))$-halver:
for each $y\in[ku]$, we simply fix one subset $S\subseteq X$ with $|S|\le k$ and $\SUM(S)=y$, if exists, and add it to
 the collection $\SSS$, and then apply Lemma~\ref{lem:discrep} to this collection.  However, the construction time is too large
(we would have to solve the subset sum problem in order to construct the collection $\SSS$ in the first place!).  
To improve the construction time, we worsen the additive error of our halver from $\OO(\sqrt{k})$ to $\OO(k/\sqrt{b})$ for a parameter $b\le k$.
We accomplish this by building a collection of small ``canonical'' subsets of sizes at most $b$, generated
from Koiliaris and Xu's algorithm, and then apply Lemma~\ref{lem:discrep} to construct a halver for this collection:


\begin{lemma}[Canonical subset generation]\label{lem:collect}
Given a set $X$ of integers in $[u]$ and a number $b$,
we can generate a collection $\SSS$ of $\OO(b^2u)$ ``canonical'' subsets of $X$ each of size at most $b$, in $\OO(b^3u)$ deterministic time, satisfying the following property:
\begin{quote}
For every subset $S\subseteq X$,
there exists a subset $\hatS\subseteq X$ such that $|\hatS|=|S|$, $\SUM(\hatS)=\SUM(S)$, and 
$\hatS$ is expressible as a union of $O((|S|/b + 1)\log u)$ disjoint canonical subsets in $\SSS$.
\end{quote}
\end{lemma}
\begin{proof}
Run the algorithm from the proof of Lemma~\ref{lem:KX} with $k=b$.
For each dyadic interval $I=(a,a+\ell]$, $i\in [b]$, and $y\in \SUMS_i(X\cap I)$,
find one subset $S_I[i,y]\subseteq X\cap I$ with $|S_I[i,y]|=i$ and $\SUM(S_I[i,y])=y$, and add the subset to the collection $\SSS$.
The number of such subsets is $O(b^2 \ell)$ per $I$ (since there are $b$ choices of $i$ and $O(\ell b)$ choices of $y\in (ai,(a+\ell)i]$), each of which can be generated
in $\OO(b)$ time.  Over all dyadic intervals $I$ (which have total length $O(u\log u)$), the total number of subsets is $\OO(b^2u)$.  

To prove that $\SSS$ satisfies the stated property, consider an arbitrary
subset $S\subseteq X$.  First partition $(0,u]$ into a collection ${\cal I}(S)$ of disjoint dyadic intervals, each containing at most
$b$ points of $S$: namely, we start by visiting the root dyadic interval, and when we visit a dyadic interval containing
more than $b$ points of $S$, we recursively visit its two child intervals.
The number of dyadic intervals visited per level is $O(|S|/b+1)$.  Thus, the number
of dyadic intervals generated is $|{\cal I}(S)|\le O((|S|/b+1)\log u)$.
For each interval $I\in {\cal I}(S)$, let $i_I=|S\cap I|\in [b]$ and $y_I=\SUM(S\cap I)\in \SUMS_i(X\cap I)$.
Define $\hatS$ to be the union of $S_I[i_I,y_I]$ over all $I\in {\cal I}(S)$; these sets $S_I[i_I,y_I]\subseteq X\cap I$ are obviously disjoint because the intervals in ${\cal I}(S)$ are disjoint; and $|S_I[i_I,y_I]|=|S\cap I|$ and $\SUM(S_I[i_I,y_I])=\SUM(S\cap I)$.
Therefore, $|\hatS|=|S|$ and $\SUM(\hatS)=\SUM(S)$.
\end{proof}

\begin{lemma}[Fast halver construction]\label{lem:halver}
Given a set $X$ of integers in $[u]$ and $b\le k$,
we can compute a $(k,\OO(k/\sqrt{b}))$-halver of $X$ in $\OO(b^3 u)$ deterministic time.
\end{lemma}
\begin{proof}
Apply Lemma~\ref{lem:collect} to construct a collection $\SSS$ of $\OO(b^2u)$ canonical subsets in $\OO(b^3u)$ time.
Then 
apply Lemma~\ref{lem:discrep} to this collection $\SSS$, with $k$ replaced by $b$, to compute a partition $(X',X'')$ in $\OO(b^3u)$ time.

To prove that $(X',X'')$ satisfies the $(k,\OO(k/\sqrt{b}))$-halver property,
consider an arbitrary subset $S\subseteq X$ with $|S|\le k$.
Let $\hatS\subseteq X$ be such that $|\hatS|=|S|\le k$, $\SUM(\hatS)=\SUM(S)$, and
$\hatS$ is expressible as a union of $O((k/b)\log u)$ disjoint canonical subsets $S_j\in\SSS$.
We have $|S_j\cap X'|, |S_j\cap X''|\le  |S_j|/2 +  \OO(\sqrt{b})$ for each such subset $S_j$.
Summing over $S_j$ gives $|\hatS\cap X'|, |\hatS\cap X''|\le |\hatS|/2 + \OO(k/b\cdot\sqrt{b})$.
\end{proof}

\mysubsection{Putting everything together}

Finally, we use halvers to obtain a deterministic divide-and-conquer algorithm for subset sum (which may be viewed as a variant of Bringmann's randomized algorithm~\cite{Bringmann17}):

\begin{lemma}\label{lem:DC}
Given a set $X$ of integers in $[u]$ and a number $k$,
we can compute a set $\ANS\subseteq [ku]$ such that $\SUMS_{\le k}(X)\subseteq\ANS\subseteq \SUMS(X)$,
in $\OO(ku)$ deterministic time.
\end{lemma}
\begin{proof}
The algorithm proceeds as follows:
\begin{enumerate}
\item Compute a $(k,k/\log k)$-halver $(X',X'')$ of $X$ in $\OO(u)$ time, by  
applying Lemma~\ref{lem:halver} with $b=\log^c(ku)$ for a sufficiently large constant $c$ (assuming $k>b$).
\item
Recursively solve the problem for the input set $X'$ and number $k/2+k/\log k$;
let $\ANS'\subseteq [ku]$ be the output set.
Recursively solve the problem for the input set $X''$ and number $k/2+k/\log k$;
let $\ANS''\subseteq [ku]$ be the output set.
\item
Return the clipped sumset $\ANS = (\ANS'+\ANS'')\cap [ku]$, which can be computed by a convolution two Boolean vectors of length $O(ku)$ in
$\OO(ku)$ time.
\end{enumerate}
The running time satisfies the recurrence: 
\[ T(k,u) \ =\ 2\,T(k/2+k/\log k,u) + \OO(ku).
\]
For the base case $k\le b$ (which implies $k\le \log^{O(1)} u$), we can use Lemma~\ref{lem:KX} to compute $\SUMS_{\le k}(X)$ in $\OO(u)$ time and
thus get $T(k,u)=\OO(u)$.
The recurrence solves to $T(k,u)=\OO(ku)$.

To prove correctness, consider a subset $S\subseteq X$ with $|S|\le k$ and $\SUM(S)=y$.
By definition of halver, there exists a subset $\hatS\subseteq X$ with $|\hatS|\le k$,
$\SUM(\hatS)=y$, and $|\hatS\cap X'|, |\hatS\cap X''|\le k/2+k/\log k$.
Then $\SUM(\hatS\cap X')\in \ANS'$ and $\SUM(\hatS\cap X'')\in \ANS''$ by induction,
and so $y=\SUM(\hatS\cap X') + \SUM(\hatS\cap X'')\in \ANS$.
\end{proof}

\begin{theorem}
Given a set $X$ of $n$ integers in $[t]$,
we can compute $\SUMS(X)\cap [t]$ in $\OO(t)$ deterministic time.
\end{theorem}
\begin{proof}
For each $i\in [\lceil\log t\rceil]$, let $X_i=X\cap (t/2^i, t/2^{i-1}]$. 
Compute $\SUMS(X_i)\cap [t]$ by applying Lemma~\ref{lem:DC}
with $u=t/2^{i-1}$ and $k=2^i$ in $\OO(2^i\cdot t/2^{i-1})=\OO(t)$ time (this is because
$\SUMS(X_i)\cap [t] = \SUMS_{\le 2^i}(X_i)\cap [t]$ when $X_i\subseteq (t/2^i,\infty)$).
We return the clipped sumset $((\SUMS(X_0)\cap[t]) + (\SUMS(X_1)\cap[t]) + \cdots)\cap [t]$,
which can be computed by $O(\log t)$ convolutions of Boolean vectors of length $O(t)$ in
$\OO(t)$ time.
\end{proof}


\mysection{Output-Sensitive Algorithms}\label{sec:os}

In this section, we describe another application of our technique: 
the derandomization of a recent output-sensitive algorithm for all-targets subset sum
by Bringmann, Fischer, and Nakos~\cite{BringmannFN25}, which runs in 
$\OO(|\SUMS(X)\cap [t]|\sqrt{n})$ time.
An earlier output-sensitive algorithm by Bringmann and Nakos~\cite{BringmannN20}, which runs in
$\OO(|\SUMS(X)\cap[t]|^{4/3})$ time, can also be derandomized in a similar way, as we will briefly mention at the end.

\mysubsection{Preliminaries}

We begin with a known subroutine on output-sensitive computation of sumsets.  This subproblem was first solved
by Cole and Hariharan~\cite{ColeH02} with randomization, which was later derandomized/improved in a series of subsequent papers
\cite{ChanL15,BringmannFN21,BringmannFN22,JinX24}:

\begin{lemma}\label{lem:sumset}
Given two sets $A$ and $B$ of integers in $[t]$,
we can compute the sumset $A+B$ in $\OO(|A+B|)$ 
deterministic time.
\end{lemma}

Output-sensitive computation of clipped sumsets is more expensive.  Bringmann and Nakos~\cite{BringmannN20}
described a simple approach yielding the following (using the subroutine in Lemma~\ref{lem:sumset}):

\begin{lemma}\label{lem:clipped:sumset}
Given two sets $A$ and $B$ of integers in $[t]$,
we can compute the clipped sumset $C = (A+B)\cap [t]$ in $\OO(\sqrt{|A||B||C|})\le \OO((|A|+|C|)\sqrt{|B|})$ 
deterministic time.
\end{lemma}

The next lemma below, observed by Bringmann, Fischer, and Nakos~\cite[Lemma~3.1]{BringmannFN25} (but stated in a different form),
derandomizes the small $k$ case using Bringmann's color-coding approach~\cite{Bringmann17}.
For completeness, we include a quick proof sketch:

\begin{lemma}\label{lem:color:code}
Given a set $X$ of $n$ integers in $[t]$, a set $A$ of integers in $[t]$, and a number $k$,
we can compute a set $\ANS$ such that $(A+\Sigma_{\le k}(X))\cap[t]\subseteq\ANS\subseteq(A+\Sigma(X))\cap[t]$, 
in $\OO(k^{O(1)}|(A+\SUMS(X))\cap[t]|\sqrt{n})$ deterministic time. 

Afterwards, for any given $y\in \ANS$, we can report one number $a\in A$ and one
subset $S\subseteq X$ with $a+\SUM(S)=y$, in $\OO(|S|)$ time.
\end{lemma}
\begin{proof}
There exists a family ${\cal H}$ of $O(k^{O(1)}\log n)$ ``perfect''
hash functions from $X$ to $\{1,\ldots,k^2\}$ satisfying the following property:
for every subset $S\subseteq X$ with $|S|\le k$, there exists $h\in{\cal H}$ such that
$\{h(x): x\in S\}$ are all distinct.  The family can be constructed in
$O(k^{O(1)}\log^{O(1)}n)$ time.  E.g., see \cite[Section~4]{AlonYZ95}.

Fix $h\in{\cal H}$.
For each $j\in \{1,\ldots,k^2\}$, let $X_{h,j}=\{x\in X: h(x)=j\}$.
Compute the clipped sumset
\[ \ANS_h \:=\: (A + (X_{h,1}\cup\{0\}) + \cdots + (X_{h,k^2}\cup\{0\})) \cap [t]
\]
by applying Lemma~\ref{lem:clipped:sumset} $O(k^2)$ times in
$\OO(k^{O(1)}|(A+\SUMS(X))\cap[t]|\sqrt{n})$ time.
Then $\ANS=\bigcup_{h\in{\cal H}}\ANS_h$ fulfills the output requirement.
For any given $y\in \ANS$, we can report a corresponding number $a\in A$ and 
subset $S\subseteq X$ again using the deterministic
witness-finding technique by Alon and Naor~\cite{AlonN96}.
%
\end{proof}

\mysubsection{Halver}

As before, the key idea is the construction of a halver.  We now need a ``weighted'' variant
that halves the sum of a subset rather than its cardinality:

\begin{definition}
Given a set $X$ of positive integers, and numbers $t$ and $\Delta$,
a \emph{weighted $(t,\Delta)$-halver} of $X$ is a partition of $X$ into two disjoint subsets $X'$ and $X''$ with the following property:
\begin{quote}
For every subset $S\subseteq X$ with $\SUM(S)\le t$, 
there exists a subset $\hatS\subseteq X$ with $\SUM(\hatS)=\SUM(S)$,
and $\SUM(\hatS\cap X'),\SUM(\hatS\cap X'')\le \SUM(\hatS)/2 + \Delta$.
\end{quote}
\end{definition}

As before, we will generate a collection of small canonical subsets,
to which we will then apply Lemma~\ref{lem:discrep} to construct a halver.
Unforunately, Koiliaris and Xu's algorithm is no longer applicable, since it is not
output-sensitive.  Instead, we will use the algorithm from Lemma~\ref{lem:color:code} for small $k$ for this
purpose.
However, it is not clear how canonical subsets can be obtained from this algorithm, since it does
not work via divide-and-conquer or dyadic intervals.
Our new idea is to invoke the algorithm for Lemma~\ref{lem:color:code} multiple times on different
inputs associated with different dyadic intervals.  Bounding the total cost of these multiple
invocations may at first appear problematic, but luckily the following simple observation comes to the rescue:

\begin{observation}\label{obs}
Let $X$ be a sorted list of numbers.
Suppose $X$ is partitioned into contiguous sublists $X_1,\ldots,X_m$.
Then $|\SUMS_{\le k}(X_1)\cap [t]|+\cdots + |\SUMS_{\le k}(X_m)\cap [t]| \le O(k)|\SUMS(X)\cap[t]|$.
\end{observation}
\begin{proof}
For a fixed $i$, $\SUMS_i(X_1),\ldots,\SUMS_i(X_m)$ are disjoint,
since every element in $\SUMS_i(X_j)$ is smaller than every element in $\SUMS_i(X_{j+1})$ for every $j$.
Thus, $|\SUMS_i(X_1)\cap [t]|+\cdots + |\SUMS_i(X_m)\cap [t]| \le |\SUMS_i(X)\cap[t]|$.
The observation then follows by summing over all $i\in [k]$.
\end{proof}

\begin{lemma}[Canonical subset generation]\label{lem:collect:os}
Given a constant $\delta\in (0,1]$, a set $X$ of $n$ integers in $[u/(1+\delta),u]$, and numbers $t$ and $b$,
we can generate a collection $\SSS$ of $\OO(b|\SUMS(X)\cap [t]|)$ ``canonical'' subsets of $X$ each of size at most $(1+\delta)b$, in $\OO(b^{O(1)}|\SUMS(X)\cap [t]|\sqrt{n})$ deterministic time, satisfying the following property:
\begin{quote}
For every subset $S\subseteq X$ with $\SUM(S)\le t$,
there exists a subset $\hatS\subseteq X$ such that $\SUM(\hatS)=\SUM(S)$ and 
$\hatS$ is expressible as a union of $O((|S|/b + 1)\log n)$ disjoint canonical subsets in $\SSS$.
\end{quote}
\end{lemma}
\begin{proof}
W.l.o.g., assume that $n$ is a power of 2.  Sort the elements $x_1,\ldots,x_n$ of $X$.
For each dyadic interval $I=(a,a+\ell]\subseteq (0,n]$, run the algorithm in Lemma~\ref{lem:color:code}
with $A=\emptyset$, $k=b$, and $t$ replaced by $\min\{t,bu\}$, on the input set $X_I := \{x_j: j\in I\}$, to get output set $\ANS_I$.
For each dyadic interval $I$ and each $y\in \ANS_I\subseteq\SUMS(X_I)\cap[\min\{t,bu\}]$,
find one subset $S_I[y]\subseteq X_I$ with $\SUM(S_I[y])=y$, and add the subset to the collection $\SSS$.
Note that $|S_I[y]|\le bu/(u/(1+\delta))=(1+\delta)b$.
By Observation~\ref{obs}, the sum of $|\SUMS(X_I)\cap [\min\{t,bu\}]|=|\SUMS_{\le (1+\delta)b}(X_I)\cap [\min\{t,bu\}]|$ over all dyadic intervals $I$ of a fixed length
is at most $O(b|\SUMS(X)\cap [\min\{t,bu\}]|)$; thus, the sum over all dyadic intervals $I$ of all lengths
is $O(b|\SUMS(X)\cap [\min\{t,bu\}]|\log n)$.  
So, the total number of subsets in $\SSS$ is bounded by $O(b|\SUMS(X)\cap [t]|\log n)$,
and the computation time is $\OO(b^{O(1)}|\SUMS(X)\cap [t]|\sqrt{n})$.

To prove that $\SSS$ satisfies the stated property, consider an arbitrary
subset $S\subseteq X$. First partition $(0,n]$ into a collection ${\cal I}(S)$ of disjoint dyadic intervals, each containing at most
$b$ indices of $\{j:x_j\in S\}$: namely, we start by visiting the root dyadic interval, and when we visit a dyadic interval containing
more than $b$ indices of $\{j:x_j\in S\}$, we recursively visit its two child intervals.
The number of dyadic intervals visited per level is $O(|S|/b+1)$.  Thus, the number
of dyadic intervals generated is $|{\cal I}(S)|\le O((|S|/b+1)\log n)$.
For each interval $I\in {\cal I}(S)$, let $y_I=\SUM(S\cap X_I)\in \SUMS_{\le b}(X_I)\cap[\min\{t,bu\}]\subseteq \ANS_I$.
Define $\hatS$ to be the union of $S_I[y_I]$ over all $I\in {\cal I}(S)$; these sets $S_I[y_I]\subseteq X_I$ are obviously disjoint because the intervals in ${\cal I}(S)$ are disjoint; and 
$\SUM(S_I[y_I])=\SUM(S\cap X_I)$.
Therefore, $\SUM(\hatS)=\SUM(S)$.
\end{proof}

\begin{lemma}[Fast halver construction]\label{lem:halver:os}
Given a set $X$ of $n$ integers in $[u]$ and numbers $t$ and $b$,
we can compute a weighted $(t,\OO(t/\sqrt{b} + bu))$-halver of $X$ in $\OO(b^{O(1)}|\SUMS(X)\cap[t]|\sqrt{n})$ deterministic time.
\end{lemma}
\begin{proof}
Let $\delta=1/\sqrt{b}$.
For each $\ell\in [\lceil\log_{1+\delta}u\rceil]$,
let $X^{(\ell)}=X\cap [(1+\delta)^\ell,(1+\delta)^{\ell+1})$.
Apply Lemma~\ref{lem:collect:os} to the input set $X^{(\ell)}$ (with $u$ replaced by $(1+\delta)^{\ell+1}$)
to construct a collection $\SSS^{(\ell)}$ of $\OO(b|\SUMS(X)\cap[t]|)$ canonical subsets
in $\OO(b^{O(1)}|\SUMS(X)\cap[t]|\sqrt{n})$ time.
Let $\SSS=\bigcup_{\ell}\SSS^{(\ell)}$.
Then apply Lemma~\ref{lem:discrep} to this collection $\SSS$, with $k=(1+\delta)b$,
to compute a partition $(X',X'')$ in $\OO(b^{O(1)}|\SUMS(X)\cap[t]|)$ time.

To prove that $(X',X'')$ satisfies the weighted $(t,\OO(t/\sqrt{b}+bu))$-halver property,
consider an arbitrary subset $S\subseteq X$ with $\SUM(S)\le t$.
Let $\hatS^{(\ell)}\subseteq X^{(\ell)}$ be such that  $\SUM(\hatS^{(\ell)})=\SUM(S\cap X^{(\ell)})$ and
$\hatS^{(\ell)}$ is expressible as a union of $O((|S\cap X^{(\ell)}|/b + 1)\log n)$ disjoint canonical subsets $S_j^{(\ell)}\in\SSS^{(\ell)}$.
We have $|S_j^{(\ell)}\cap X'|, |S_j^{(\ell)}\cap X''|\le  |S_j^{(\ell)}|/2 + \OO(\sqrt{b})$ for each such subset $S_j$.
Summing over $S_j^{(\ell)}$ gives 
\begin{eqnarray*}
 \max\{|\hatS^{(\ell)}\cap X'|, |\hatS^{(\ell)}\cap X''|\} &\le& |\hatS^{(\ell)}|/2 + \OO((|S\cap X^{(\ell)}|/b+1)\cdot\sqrt{b})\\[2pt]
  &=& |\hatS^{(\ell)}|/2 + \OO(|S\cap X^{(\ell)}|/\sqrt{b}+\sqrt{b}).
\end{eqnarray*}
Let $\hatS=\bigcup_\ell \hatS^{(\ell)}$.  Then $\SUM(\hatS)=\SUM(S)$.
Now, $\frac{1}{1+\delta}\SUM(\hatS\cap X')\le \sum_\ell (1+\delta)^\ell |\hatS^{(\ell)}\cap X'|$,
and $\frac{1}{1+\delta}\SUM(\hatS\cap X'')\le \sum_\ell (1+\delta)^\ell |\hatS^{(\ell)}\cap X''|$.
Thus, 
\begin{eqnarray*}
 \max\{\SUM(\hatS\cap X'),\SUM(\hatS\cap X'')\} &\le& 
 \sum_\ell (1+\delta)^\ell \left(|\hatS^{(\ell)}|/2 + \OO(|S\cap X^{(\ell)}|/\sqrt{b}+\sqrt{b})\right) + O(\delta t)\\
 &\le &
 \SUM(\hatS)/2 +   \OO\left(\SUM(S)/\sqrt{b} +
 \sum_\ell (1+\delta)^\ell \sqrt{b}\right) + O(\delta t)\\
 & =& \SUM(\hatS)/2 +   \OO(t/\sqrt{b} 
 + \sqrt{b}u/\delta + \delta t)\ =\ \SUM(\hatS)/2 + \OO(t/\sqrt{b} + bu).\PAPER{\ \ \ \ \ }
\end{eqnarray*}

\ \PAPER{\par\vspace{-2\bigskipamount}}
\end{proof}

\mysubsection{Putting everything together}

We are now ready to derandomize Bringmann, Fischer, and Nakos' algorithm~\cite{BringmannFN25}.
The algorithm is a multi-way divide-and-conquer, so we first extend halvers to split into $r$ parts:

\begin{definition}
Given a set $X$ of positive integers, and numbers $t$ and $\Delta$,
a \emph{weighted $r$-way $(t,\Delta)$-splitter} of $X$ is a partition of $X$ into $r$ disjoint subsets $X_1,\ldots,X_r$  with the following property:
\begin{quote}
For every subset $S\subseteq X$ with $\SUM(S)\le t$, 
there exists a subset $\hatS\subseteq X$ with $\SUM(\hatS)=\SUM(S)$,
and $\SUM(\hatS\cap X_j)\le \SUM(\hatS)/r + \Delta$ for each $j\in\{1,\ldots,r\}$.
\end{quote} 
\end{definition}

\begin{corollary}[Fast $r$-way splitter construction]\label{cor:splitter:os}
Given a set $X$ of $n$ integers in $[u]$ and numbers $t,b,r$,
we can compute a weighted $r$-way $(t,\OO((t/\sqrt{b} + bu)\log r))$-splitter of $X$ in $\OO(rb^{O(1)}|\SUMS(X)\cap[t]|\sqrt{n})$ deterministic time.
\end{corollary}
\begin{proof}
Just apply Lemma~\ref{lem:halver:os} recursively for $\log r$ levels.
\end{proof}

\begin{theorem}
Given a set $X$ of $n$ integers in $[t]$,
we can compute $\SUMS(X)\cap[t]$ in $\OO(|\SUMS(X)\cap[t]|\sqrt{n})$ deterministic time.
\end{theorem}
\begin{proof}
Let $u_S= \lceil t/(r^6\log^c t)\rceil$ for a sufficiently large constant $c$.
Let $X_S=X\cap [u_S]$ and $X_L=X\cap (u_S,t]$.
The following is a reinterpretation of 
Bringmann, Fischer, and Nakos' algorithm~\cite{BringmannFN25}:
\begin{enumerate}
\item Compute a weighted $r$-way $(t,t/r^2)$-splitter $(X_1,\ldots,X_r)$ of $X_S$ 
in $\OO(r^{O(1)}|\SUMS(X)\cap[t]|\sqrt{n})$ deterministic time, by applying
Corollary~\ref{cor:splitter:os} with $b=\OO(r^4)$ and $u$ replaced by $u_S$ (the parameters are chosen so that
 $(t/\sqrt{b}+bu_S)\log r \ll t/r^2$).
\item Recursively compute $A_j=\SUMS(X_j)\cap [t/r+t/r^2]$ for each $j\in\{1,\ldots,r\}$.
\item Compute the sumset $A=A_1+\cdots+A_r$ by applying Lemma~\ref{lem:sumset} $O(r)$ times in
$\OO(r|A|)$ deterministic time.
\item Return $(A+\SUMS(X_L))\cap [t]$, which can be computed by applying Lemma~\ref{lem:color:code}
with $k=r^6\log^c t$ in $\OO(r^{O(1)}|\SUMS(X)\cap[t]|\sqrt{n})$ deterministic time (this is because
$\SUMS(X_L)\cap [t] = \SUMS_{\le r^6\log^ct}(X_L)\cap [t]$ when $X_L\subseteq (t/(r^6\log^c t),\infty)$).
\end{enumerate}

Bringmann, Fischer, and Nakos~\cite{BringmannFN25} proved that
\[
|(\SUMS(X_1)\cap [t/r(1+1/r)])+\cdots + (\SUMS(X_r)\cap [t/r(1+1/r)])|\:\le\: |\SUMS(X_1\cup\cdots\cup X_r)\cap[t]|^{1+1/(r-1)}
\]
by using a ``submultiplicativity'' property for sumsets
due to Gyarmati, Matolcsi, and Ruzsa~\cite{GyarmatiMR10}.
Thus, $|A|=|A_1+\cdots+A_r|\le  |\SUMS(X)\cap [t]|^{1+1/(r-1)}$.
Bringmann, Fischer, and Nakos further noted that $|A_1|+\cdots+|A_r|\le |A_1+\cdots+A_r|+r-1$.

Thus, the running time for $\mu=|\SUMS(X)\cap[t]|$ satisfies the recurrence
\begin{eqnarray*}
 T(n,t,\mu) &\le& \!\!\max_{\scriptsize\begin{array}{c}n_1,\ldots,n_r,\mu_1,\ldots,\mu_r:\\ n_1+\cdots+n_r\le n,\\
 \mu_1+\cdots+\mu_r\le \mu^{1+1/(r-1)}+r-1\end{array}} \!\!\!\!\!\!\!\!\!\!
\Big(T(n_1,t/r+t/r^2,\mu_1)+\cdots+T(n_r,t/r+t/r^2,\mu_r)\\[-5ex]
 && \qquad\qquad\qquad\qquad\qquad\qquad\qquad\qquad\qquad\qquad {} + \OO(r^{O(1)}\cdot (\mu^{1+1/(r-1)} + \mu\sqrt{n}))\Big).
\end{eqnarray*}
By choosing $r=\log^2 t + 1$, we have $\mu^{1/(r-1)}\le 2^{\log \mu/\log^2 t}\le 2^{1/\log t} = 1+ O(1/\log t)$, and
the recurrence solves to $T(n,t,\mu)=\OO(\mu\sqrt{n})$.
\end{proof}

To close, we mention that another output-sensitive subset sum algorithm by 
Bringmann and Nakos~\cite{BringmannN20} can be similarly derandomized.  They described a (more complicated) alternative
to Lemma~\ref{lem:clipped:sumset} on output-sensitive construction of clipped sumsets, running in
$\OO(|C|^{4/3})$ time (their construction is deterministic, if a deterministic version of Lemma~\ref{lem:sumset}
is used).  By replacing Lemma~\ref{lem:clipped:sumset}, keeping the rest of the algorithm the same,
and straightforwardly modifying the above analysis, we immediately obtain:

\begin{theorem}
Given a set $X$ of $n$ integers in $[t]$,
we can compute $\SUMS(X)\cap[t]$ in $\OO(|\SUMS(X)\cap[t]|^{4/3})$ deterministic time.
\end{theorem}

\mysection{Remarks}

In the $\OO(t)$-time algorithm in Section~\ref{sec:linear},
we have not tried to optimize (and thus have not written out) the number of $\log t$ factors.
It is larger than
the previous randomized algorithms (Bringmann's algorithm~\cite{Bringmann17} had 4 log factors, and
Jin and Wu's algorithm~\cite{JinW19} had 2), but is not gigantic (probably below 20).

In the output-sensitive algorithms in Section~\ref{sec:os},
we reiterate that the hidden extra factors are polylogarithmic in $n$ and $t$.
The previous randomized algorithms by Bringmann, Fischer, and Nakos~\cite{BringmannFN25} actually got factors that are
polylogarithmic in the output size $|\SUMS(X)\cap [t]|$ rather than in $t$; this seems harder to do from the derandomization perspective (but would not make a big difference except in some extreme cases).

Like previous work, our algorithms easily adapt to the setting when $X$ is a multiset---a subset of $X$ cannot use an element more times than its original multiplicity.  (Note that $n$ could be larger than $t$ here, and to avoid an
$O(n)$ term in the running time, we may need to assume that $X$ is represented as a list of distinct elements
with given multiplicities.)

In the $\OO(t)$-time algorithm in Section~\ref{sec:linear}, we generate canonical subsets by adapting
Koiliaris and Xu's algorithm~\cite{KoiliarisX19}, but in Section~\ref{sec:os}, we have realized that Koiliaris and Xu's algorithm
is not essential to our technique (though it is convenient, since it already incorporates dyadic intervals).
This realization may potentially help in future applications.

We hope that our technique may find many further applications in derandomizing other algorithms in this area, since a number of previous algorithms used randomized partitioning similar to halvers---to name just two examples, Chen, Lian, Mao, and Zhang's $\OO(n+ \sqrt{ut})$-time subset sum algorithm~\cite{ChenLMZ24} (where $u=\max(X)$) and Bringmann, D\"urr, and Polak's knapsack algorithm~\cite{BringmannDP24}.  We leave the question
of their derandomization to future work (these algorithms have other components that may pose additional challenges
to derandomize).

\PAPER{\myparagraph{Acknowledgement.}}
\SIAMV{\section*{Acknowledgement.}}
I thank Karl Bringmann for suggesting potential applications of the technique, including derandomization
of the output-sensitive subset sum algorithms.

{\small
\bibliographystyle{alphaurl}
\bibliography{det_subsetsum}

@article{CyganMWW19,
  author       = {Marek Cygan and
                  Marcin Mucha and
                  Karol W{\k e}grzycki and
                  Micha{\l} W{\l}odarczyk},
  title        = {On Problems Equivalent to {$(\min, +)$-}Convolution},
  journal      = {{ACM} Trans. Algorithms},
  volume       = {15},
  number       = {1},
  pages        = {14:1--14:25},
  year         = {2019},
  url          = {https://doi.org/10.1145/3293465},
  doi          = {10.1145/3293465},
  bibsource    = {dblp computer science bibliography, https://dblp.org}
}

@inproceedings{BringmannFN25,
  author       = {Karl Bringmann and
                  Nick Fischer and
                  Vasileios Nakos},
  _editor       = {Yossi Azar and
                  Debmalya Panigrahi},
  title        = {Beating {B}ellman's Algorithm for Subset Sum},
  booktitle    = {Proc. 36th Annual ACM--SIAM Symposium on Discrete Algorithms (SODA)},
  pages        = {4596--4612},
  _publisher    = {{SIAM}},
  year         = {2025},
  _url          = {https://doi.org/10.1137/1.9781611978322.157},
  doi          = {10.1137/1.9781611978322.157},
  timestamp    = {Tue, 28 Jan 2025 14:38:41 +0100},
  biburl       = {https://dblp.org/rec/conf/soda/BringmannFN25.bib},
  bibsource    = {dblp computer science bibliography, https://dblp.org}
}

@inproceedings{BringmannN20,
  author       = {Karl Bringmann and
                  Vasileios Nakos},
  _editor       = {Konstantin Makarychev and
                  Yury Makarychev and
                  Madhur Tulsiani and
                  Gautam Kamath and
                  Julia Chuzhoy},
  title        = {Top-$k$-convolution and the quest for near-linear output-sensitive subset
                  sum},
  booktitle    = {Proc. 52nd Annual {ACM} Symposium on Theory
                  of Computing (STOC)},
  pages        = {982--995},
  _publisher    = {{ACM}},
  year         = {2020},
  _url          = {https://doi.org/10.1145/3357713.3384308},
  doi          = {10.1145/3357713.3384308},
  timestamp    = {Mon, 18 Dec 2023 07:33:37 +0100},
  biburl       = {https://dblp.org/rec/conf/stoc/BringmannN20.bib},
  bibsource    = {dblp computer science bibliography, https://dblp.org}
}

@inproceedings{Bringmann17,
  author       = {Karl Bringmann},
  _editor       = {Philip N. Klein},
  title        = {A Near-Linear Pseudopolynomial Time Algorithm for Subset Sum},
  booktitle    = {Proc. 28th Annual ACM--SIAM Symposium on Discrete
                  Algorithms (SODA)},
  pages        = {1073--1084},
  _publisher    = {{SIAM}},
  year         = {2017},
  _url          = {https://doi.org/10.1137/1.9781611974782.69},
  doi          = {10.1137/1.9781611974782.69},
  timestamp    = {Tue, 02 Feb 2021 17:07:33 +0100},
  biburl       = {https://dblp.org/rec/conf/soda/Bringmann17.bib},
  bibsource    = {dblp computer science bibliography, https://dblp.org}
}

@article{KoiliarisX19,
  author       = {Konstantinos Koiliaris and
                  Chao Xu},
  title        = {Faster Pseudopolynomial Time Algorithms for Subset Sum},
  journal      = {{ACM} Trans. Algorithms},
  volume       = {15},
  number       = {3},
  pages        = {40:1--40:20},
  year         = {2019},
  _url          = {https://doi.org/10.1145/3329863},
  doi          = {10.1145/3329863},
  timestamp    = {Sat, 08 Jan 2022 02:22:07 +0100},
  biburl       = {https://dblp.org/rec/journals/talg/KoiliarisX19.bib},
  bibsource    = {dblp computer science bibliography, https://dblp.org},
note = {Preliminary version in SODA 2017}
}

@article{KoiliarisX18,
  author       = {Konstantinos Koiliaris and
                  Chao Xu},
  title        = {Subset Sum Made Simple},
  journal      = {CoRR},
  volume       = {abs/1807.08248},
  year         = {2018},
  _url          = {http://arxiv.org/abs/1807.08248},
  eprinttype    = {arXiv},
  eprint       = {1807.08248},
  timestamp    = {Mon, 13 Aug 2018 16:47:38 +0200},
  biburl       = {https://dblp.org/rec/journals/corr/abs-1807-08248.bib},
  bibsource    = {dblp computer science bibliography, https://dblp.org}
}

@inproceedings{JinW19,
  author       = {Ce Jin and
                  Hongxun Wu},
  _editor       = {Jeremy T. Fineman and
                  Michael Mitzenmacher},
  title        = {A Simple Near-Linear Pseudopolynomial Time Randomized Algorithm for
                  Subset Sum},
  booktitle    = {Proc. 2nd Symposium on Simplicity in Algorithms (SOSA)},
  _series       = {OASIcs},
  _volume       = {69},
  pages        = {17:1--17:6},
  _publisher    = {Schloss Dagstuhl - Leibniz-Zentrum f{\"{u}}r Informatik},
  year         = {2019},
  _url          = {https://doi.org/10.4230/OASIcs.SOSA.2019.17},
  doi          = {10.4230/OASICS.SOSA.2019.17},
  timestamp    = {Mon, 03 Mar 2025 21:22:27 +0100},
  biburl       = {https://dblp.org/rec/conf/soda/JinW19.bib},
  bibsource    = {dblp computer science bibliography, https://dblp.org}
}

@inproceedings{ChenLMZ24,
  author       = {Lin Chen and
                  Jiayi Lian and
                  Yuchen Mao and
                  Guochuan Zhang},
  title        = {An Improved Pseudopolynomial Time Algorithm for Subset Sum},
  booktitle    = {Proc. 65th {IEEE} Annual Symposium on Foundations of Computer Science (FOCS)},
  pages        = {2202--2216},
  _publisher    = {{IEEE}},
  year         = {2024},
  _url          = {https://doi.org/10.1109/FOCS61266.2024.00129},
  doi          = {10.1109/FOCS61266.2024.00129},
  timestamp    = {Tue, 10 Dec 2024 07:54:49 +0100},
  biburl       = {https://dblp.org/rec/conf/focs/ChenL0Z24.bib},
  bibsource    = {dblp computer science bibliography, https://dblp.org}
}

@inproceedings{ChenLMZ24a,
  author       = {Lin Chen and
                  Jiayi Lian and
                  Yuchen Mao and
                  Guochuan Zhang},
  _editor       = {Bojan Mohar and
                  Igor Shinkar and
                  Ryan O'Donnell},
  title        = {Approximating Partition in Near-Linear Time},
  booktitle    = {Proc. 56th Annual {ACM} Symposium on Theory of Computing (STOC)},
  pages        = {307--318},
  _publisher    = {{ACM}},
  year         = {2024},
  _url          = {https://doi.org/10.1145/3618260.3649727},
  doi          = {10.1145/3618260.3649727},
  timestamp    = {Fri, 16 Aug 2024 09:48:26 +0200},
  biburl       = {https://dblp.org/rec/conf/stoc/ChenLMZ24.bib},
  bibsource    = {dblp computer science bibliography, https://dblp.org}
}

@inproceedings{LokshtanovN10,
  author       = {Daniel Lokshtanov and
                  Jesper Nederlof},
  _editor       = {Leonard J. Schulman},
  title        = {Saving space by algebraization},
  booktitle    = {Proc. 42nd {ACM} Symposium on Theory of Computing (STOC)},
  pages        = {321--330},
  _publisher    = {{ACM}},
  year         = {2010},
  _url          = {https://doi.org/10.1145/1806689.1806735},
  doi          = {10.1145/1806689.1806735},
  timestamp    = {Sun, 19 Jan 2025 13:28:36 +0100},
  biburl       = {https://dblp.org/rec/conf/stoc/LokshtanovN10.bib},
  bibsource    = {dblp computer science bibliography, https://dblp.org}
}

@inproceedings{PolakRW21,
  author       = {Adam Polak and
                  Lars Rohwedder and
                  Karol W{\k e}grzycki},
  _editor       = {Nikhil Bansal and
                  Emanuela Merelli and
                  James Worrell},
  title        = {Knapsack and Subset Sum with Small Items},
  booktitle    = {Proc. 48th International Colloquium on Automata, Languages, and Programming (ICALP)},
  _series       = {LIPIcs},
  _volume       = {198},
  pages        = {106:1--106:19},
  _publisher    = {Schloss Dagstuhl - Leibniz-Zentrum f{\"{u}}r Informatik},
  year         = {2021},
  _url          = {https://doi.org/10.4230/LIPIcs.ICALP.2021.106},
  doi          = {10.4230/LIPICS.ICALP.2021.106},
  timestamp    = {Wed, 21 Aug 2024 22:46:00 +0200},
  biburl       = {https://dblp.org/rec/conf/icalp/0001RW21.bib},
  bibsource    = {dblp computer science bibliography, https://dblp.org}
}

@inproceedings{BringmannW21,
  author       = {Karl Bringmann and
                  Philip Wellnitz},
  _editor       = {D{\'{a}}niel Marx},
  title        = {On Near-Linear-Time Algorithms for Dense Subset Sum},
  booktitle    = {Proc. 32nd ACM--SIAM Symposium on Discrete Algorithms (SODA)},
  pages        = {1777--1796},
  _publisher    = {{SIAM}},
  year         = {2021},
  _url          = {https://doi.org/10.1137/1.9781611976465.107},
  doi          = {10.1137/1.9781611976465.107},
  timestamp    = {Thu, 15 Jul 2021 13:48:57 +0200},
  biburl       = {https://dblp.org/rec/conf/soda/BringmannW21.bib},
  bibsource    = {dblp computer science bibliography, https://dblp.org}
}

@inproceedings{BringmannFN22,
  author       = {Karl Bringmann and
                  Nick Fischer and
                  Vasileios Nakos},
  _editor       = {Joseph (Seffi) Naor and
                  Niv Buchbinder},
  title        = {Deterministic and {L}as {V}egas Algorithms for Sparse Nonnegative Convolution},
  booktitle    = {Proc. 33rd ACM--SIAM Symposium on Discrete Algorithms (SODA)},
  pages        = {3069--3090},
  _publisher    = {{SIAM}},
  year         = {2022},
  _url          = {https://doi.org/10.1137/1.9781611977073.119},
  doi          = {10.1137/1.9781611977073.119},
  timestamp    = {Tue, 12 Apr 2022 11:24:56 +0200},
  biburl       = {https://dblp.org/rec/conf/soda/BringmannFN22.bib},
  bibsource    = {dblp computer science bibliography, https://dblp.org}
}

@inproceedings{BringmannFN21,
  author       = {Karl Bringmann and
                  Nick Fischer and
                  Vasileios Nakos},
  _editor       = {Samir Khuller and
                  Virginia Vassilevska Williams},
  title        = {Sparse nonnegative convolution is equivalent to dense nonnegative
                  convolution},
  booktitle    = {Proc. 53rd Annual {ACM} Symposium on Theory of Computing (STOC)},
  pages        = {1711--1724},
  _publisher    = {{ACM}},
  year         = {2021},
  _url          = {https://doi.org/10.1145/3406325.3451090},
  doi          = {10.1145/3406325.3451090},
  timestamp    = {Tue, 22 Jun 2021 20:03:56 +0200},
  biburl       = {https://dblp.org/rec/conf/stoc/BringmannFN21.bib},
  bibsource    = {dblp computer science bibliography, https://dblp.org}
}

@inproceedings{JinX24,
  author       = {Ce Jin and
                  Yinzhan Xu},
  _editor       = {Bojan Mohar and
                  Igor Shinkar and
                  Ryan O'Donnell},
  title        = {Shaving Logs via Large Sieve Inequality: Faster Algorithms for Sparse
                  Convolution and More},
  booktitle    = {Proc. 56th Annual {ACM} Symposium on Theory of Computing (STOC)},
  pages        = {1573--1584},
  _publisher    = {{ACM}},
  year         = {2024},
  _url          = {https://doi.org/10.1145/3618260.3649605},
  doi          = {10.1145/3618260.3649605},
  timestamp    = {Sun, 19 Jan 2025 13:28:38 +0100},
  biburl       = {https://dblp.org/rec/conf/stoc/0001X24.bib},
  bibsource    = {dblp computer science bibliography, https://dblp.org}
}

@article{AbboudBHS22,
  author       = {Amir Abboud and
                  Karl Bringmann and
                  Danny Hermelin and
                  Dvir Shabtay},
  title        = {{SETH}-based Lower Bounds for Subset Sum and Bicriteria Path},
  journal      = {{ACM} Trans. Algorithms},
  volume       = {18},
  number       = {1},
  pages        = {6:1--6:22},
  year         = {2022},
  _url          = {https://doi.org/10.1145/3450524},
  doi          = {10.1145/3450524},
  timestamp    = {Thu, 24 Feb 2022 09:26:04 +0100},
  biburl       = {https://dblp.org/rec/journals/talg/AbboudBHS22.bib},
  bibsource    = {dblp computer science bibliography, https://dblp.org}
}

@inproceedings{AxiotisBJTW19,
  author       = {Kyriakos Axiotis and
                  Arturs Backurs and
                  Ce Jin and
                  Christos Tzamos and
                  Hongxun Wu},
  _editor       = {Timothy M. Chan},
  title        = {Fast Modular Subset Sum using Linear Sketching},
  booktitle    = {Proc. 30th Annual ACM--SIAM Symposium on Discrete
                  Algorithms (SODA)},
  pages        = {58--69},
  _publisher    = {{SIAM}},
  year         = {2019},
  _url          = {https://doi.org/10.1137/1.9781611975482.4},
  doi          = {10.1137/1.9781611975482.4},
  timestamp    = {Mon, 03 Mar 2025 21:22:25 +0100},
  biburl       = {https://dblp.org/rec/conf/soda/AxiotisBJTW19.bib},
  bibsource    = {dblp computer science bibliography, https://dblp.org}
}

@inproceedings{AxiotisBBJNTW21,
  author       = {Kyriakos Axiotis and
                  Arturs Backurs and
                  Karl Bringmann and
                  Ce Jin and
                  Vasileios Nakos and
                  Christos Tzamos and
                  Hongxun Wu},
  _editor       = {Hung Viet Le and
                  Valerie King},
  title        = {Fast and Simple Modular Subset Sum},
  booktitle    = {Proc. 4th Symposium on Simplicity in Algorithms (SOSA)},
  pages        = {57--67},
  _publisher    = {{SIAM}},
  year         = {2021},
  _url          = {https://doi.org/10.1137/1.9781611976496.6},
  doi          = {10.1137/1.9781611976496.6},
  timestamp    = {Wed, 19 Jul 2023 08:08:09 +0200},
  biburl       = {https://dblp.org/rec/conf/sosa/AxiotisBBJNTW21.bib},
  bibsource    = {dblp computer science bibliography, https://dblp.org}
}

@inproceedings{CardinalI21,
  author       = {Jean Cardinal and
                  John Iacono},
  _editor       = {Hung Viet Le and
                  Valerie King},
  title        = {Modular Subset Sum, Dynamic Strings, and Zero-Sum Sets},
  booktitle    = {Proc. 4th Symposium on Simplicity in Algorithms (SOSA)},
  pages        = {45--56},
  _publisher    = {{SIAM}},
  year         = {2021},
  _url          = {https://doi.org/10.1137/1.9781611976496.5},
  doi          = {10.1137/1.9781611976496.5},
  timestamp    = {Wed, 17 Mar 2021 13:30:03 +0100},
  biburl       = {https://dblp.org/rec/conf/sosa/CardinalI21.bib},
  bibsource    = {dblp computer science bibliography, https://dblp.org}
}

@article{CyganDLMNOPSW16,
  author       = {Marek Cygan and
                  Holger Dell and
                  Daniel Lokshtanov and
                  D{\'{a}}niel Marx and
                  Jesper Nederlof and
                  Yoshio Okamoto and
                  Ramamohan Paturi and
                  Saket Saurabh and
                  Magnus Wahlstr{\"{o}}m},
  title        = {On Problems as Hard as {CNF-SAT}},
  journal      = {{ACM} Trans. Algorithms},
  volume       = {12},
  number       = {3},
  pages        = {41:1--41:24},
  year         = {2016},
  _url          = {https://doi.org/10.1145/2925416},
  doi          = {10.1145/2925416},
  timestamp    = {Sat, 19 Oct 2019 19:12:05 +0200},
  biburl       = {https://dblp.org/rec/journals/talg/CyganDLMNOPSW16.bib},
  bibsource    = {dblp computer science bibliography, https://dblp.org}
}

@article{BremnerCDEHILPT14,
  author       = {David Bremner and
                  Timothy M. Chan and
                  Erik D. Demaine and
                  Jeff Erickson and
                  Ferran Hurtado and
                  John Iacono and
                  Stefan Langerman and
                  Mihai P{\u{a}}tra{\c{s}}cu and
                  Perouz Taslakian},
  title        = {Necklaces, Convolutions, and {$X+Y$}},
  journal      = {Algorithmica},
  volume       = {69},
  number       = {2},
  pages        = {294--314},
  year         = {2014},
  _url          = {https://doi.org/10.1007/s00453-012-9734-3},
  doi          = {10.1007/S00453-012-9734-3},
  timestamp    = {Fri, 30 Nov 2018 13:28:56 +0100},
  biburl       = {https://dblp.org/rec/journals/algorithmica/BremnerCDEHILPT14.bib},
  bibsource    = {dblp computer science bibliography, https://dblp.org}
}

@article{AlonYZ95,
  author       = {Noga Alon and
                  Raphael Yuster and
                  Uri Zwick},
  title        = {Color-Coding},
  journal      = {J. {ACM}},
  volume       = {42},
  number       = {4},
  pages        = {844--856},
  year         = {1995},
  _url          = {https://doi.org/10.1145/210332.210337},
  doi          = {10.1145/210332.210337},
  timestamp    = {Wed, 14 Nov 2018 10:35:26 +0100},
  biburl       = {https://dblp.org/rec/journals/jacm/AlonYZ95.bib},
  bibsource    = {dblp computer science bibliography, https://dblp.org}
}

@inproceedings{NederlofW21,
  author       = {Jesper Nederlof and
                  Karol W{\k e}grzycki},
  _editor       = {Samir Khuller and
                  Virginia Vassilevska Williams},
  title        = {Improving {S}chroeppel and {S}hamir's algorithm for subset sum via orthogonal
                  vectors},
  booktitle    = {Proc. 53rd Annual {ACM} Symposium on Theory of Computing (STOC)},
  pages        = {1670--1683},
  _publisher    = {{ACM}},
  year         = {2021},
  _url          = {https://doi.org/10.1145/3406325.3451024},
  doi          = {10.1145/3406325.3451024},
  timestamp    = {Tue, 16 Aug 2022 23:04:42 +0200},
  biburl       = {https://dblp.org/rec/conf/stoc/NederlofW21.bib},
  bibsource    = {dblp computer science bibliography, https://dblp.org}
}

@inproceedings{JinVW21,
  author       = {Ce Jin and
                  Nikhil Vyas and
                  Ryan Williams},
  _editor       = {D{\'{a}}niel Marx},
  title        = {Fast Low-Space Algorithms for Subset Sum},
  booktitle    = {Proc. 32nd ACM--SIAM Symposium on Discrete Algorithms (SODA)},
  pages        = {1757--1776},
  _publisher    = {{SIAM}},
  year         = {2021},
  _url          = {https://doi.org/10.1137/1.9781611976465.106},
  doi          = {10.1137/1.9781611976465.106},
  timestamp    = {Wed, 19 Jul 2023 07:58:14 +0200},
  biburl       = {https://dblp.org/rec/conf/soda/JinVW21.bib},
  bibsource    = {dblp computer science bibliography, https://dblp.org}
}

@inproceedings{BelovaCKM24,
  author       = {Tatiana Belova and
                  Nikolai Chukhin and
                  Alexander S. Kulikov and
                  Ivan Mihajlin},
  _editor       = {Timothy M. Chan and
                  Johannes Fischer and
                  John Iacono and
                  Grzegorz Herman},
  title        = {Improved Space Bounds for Subset Sum},
  booktitle    = {Proc. 32nd Annual European Symposium on Algorithms (ESA)},
  _series       = {LIPIcs},
  _volume       = {308},
  pages        = {21:1--21:17},
  _publisher    = {Schloss Dagstuhl - Leibniz-Zentrum f{\"{u}}r Informatik},
  year         = {2024},
  _url          = {https://doi.org/10.4230/LIPIcs.ESA.2024.21},
  doi          = {10.4230/LIPICS.ESA.2024.21},
  timestamp    = {Tue, 01 Apr 2025 19:06:49 +0200},
  biburl       = {https://dblp.org/rec/conf/esa/BelovaCKM24.bib},
  bibsource    = {dblp computer science bibliography, https://dblp.org}
}

@article{KellererMPS03,
  author       = {Hans Kellerer and
                  Renata Mansini and
                  Ulrich Pferschy and
                  Maria Grazia Speranza},
  title        = {An efficient fully polynomial approximation scheme for the Subset-Sum
                  Problem},
  journal      = {J. Comput. Syst. Sci.},
  volume       = {66},
  number       = {2},
  pages        = {349--370},
  year         = {2003},
  _url          = {https://doi.org/10.1016/S0022-0000(03)00006-0},
  doi          = {10.1016/S0022-0000(03)00006-0},
  timestamp    = {Tue, 16 Feb 2021 14:04:03 +0100},
  biburl       = {https://dblp.org/rec/journals/jcss/KellererMPS03.bib},
  bibsource    = {dblp computer science bibliography, https://dblp.org}
}

@inproceedings{FischerKP24,
  author       = {Nick Fischer and
                  Piotr Kaliciak and
                  Adam Polak},
  _editor       = {Venkatesan Guruswami},
  title        = {Deterministic {3SUM}-Hardness},
  booktitle    = {Proc. 15th Innovations in Theoretical Computer Science Conference (ITCS)},
  _series       = {LIPIcs},
  _volume       = {287},
  pages        = {49:1--49:24},
  _publisher    = {Schloss Dagstuhl - Leibniz-Zentrum f{\"{u}}r Informatik},
  year         = {2024},
  _url          = {https://doi.org/10.4230/LIPIcs.ITCS.2024.49},
  doi          = {10.4230/LIPICS.ITCS.2024.49},
  timestamp    = {Sun, 06 Oct 2024 21:08:19 +0200},
  biburl       = {https://dblp.org/rec/conf/innovations/FischerK024.bib},
  bibsource    = {dblp computer science bibliography, https://dblp.org}
}

@inproceedings{BringmannDP24,
  author       = {Karl Bringmann and
                  Anita D{\"{u}}rr and
                  Adam Polak},
  _editor       = {Timothy M. Chan and
                  Johannes Fischer and
                  John Iacono and
                  Grzegorz Herman},
  title        = {Even Faster Knapsack via Rectangular Monotone Min-Plus Convolution
                  and Balancing},
  booktitle    = {Proc. 32nd Annual European Symposium on Algorithms (ESA)},
  _series       = {LIPIcs},
  _volume       = {308},
  pages        = {33:1--33:15},
  _publisher    = {Schloss Dagstuhl - Leibniz-Zentrum f{\"{u}}r Informatik},
  year         = {2024},
  _url          = {https://doi.org/10.4230/LIPIcs.ESA.2024.33},
  doi          = {10.4230/LIPICS.ESA.2024.33},
  timestamp    = {Mon, 03 Mar 2025 21:03:44 +0100},
  biburl       = {https://dblp.org/rec/conf/esa/BringmannD024.bib},
  bibsource    = {dblp computer science bibliography, https://dblp.org}
}

@inproceedings{ChanX24,
  author       = {Timothy M. Chan and
                  Yinzhan Xu},
  _editor       = {Merav Parter and
                  Seth Pettie},
  title        = {Simpler Reductions from {E}xact {T}riangle},
  booktitle    = {Proc. 7th Symposium on Simplicity in Algorithms (SOSA)},
  pages        = {28--38},
  publisher    = {{SIAM}},
  year         = {2024},
  _url          = {https://doi.org/10.1137/1.9781611977936.4},
  doi          = {10.1137/1.9781611977936.4},
  timestamp    = {Wed, 10 Apr 2024 20:26:09 +0200},
  biburl       = {https://dblp.org/rec/conf/sosa/ChanX24.bib},
  bibsource    = {dblp computer science bibliography, https://dblp.org}
}

@inproceedings{ChanH20,
  author       = {Timothy M. Chan and
                  Qizheng He},
  _editor       = {Martin Farach{-}Colton and
                  Inge Li G{\o}rtz},
  title        = {Reducing {3SUM} to {Convolution-3SUM}},
  booktitle    = {Proc. 3rd Symposium on Simplicity in Algorithms (SOSA)},
  pages        = {1--7},
  _publisher    = {{SIAM}},
  year         = {2020},
  _url          = {https://doi.org/10.1137/1.9781611976014.1},
  doi          = {10.1137/1.9781611976014.1},
  timestamp    = {Sun, 04 Aug 2024 19:38:54 +0200},
  biburl       = {https://dblp.org/rec/conf/soda/ChanH20.bib},
  bibsource    = {dblp computer science bibliography, https://dblp.org}
}

@inproceedings{KunnemannPS17,
  author       = {Marvin K{\"{u}}nnemann and
                  Ramamohan Paturi and
                  Stefan Schneider},
  _editor       = {Ioannis Chatzigiannakis and
                  Piotr Indyk and
                  Fabian Kuhn and
                  Anca Muscholl},
  title        = {On the Fine-Grained Complexity of One-Dimensional Dynamic Programming},
  booktitle    = {Proc. 44th International Colloquium on Automata, Languages, and Programming (ICALP)},
  _series       = {LIPIcs},
  _volume       = {80},
  pages        = {21:1--21:15},
  _publisher    = {Schloss Dagstuhl - Leibniz-Zentrum f{\"{u}}r Informatik},
  year         = {2017},
  _url          = {https://doi.org/10.4230/LIPIcs.ICALP.2017.21},
  doi          = {10.4230/LIPICS.ICALP.2017.21},
  timestamp    = {Tue, 11 Feb 2020 15:52:14 +0100},
  biburl       = {https://dblp.org/rec/conf/icalp/KunnemannPS17.bib},
  bibsource    = {dblp computer science bibliography, https://dblp.org}
}

@article{ChanW21,
  author       = {Timothy M. Chan and
                  R. Ryan Williams},
  title        = {Deterministic {APSP}, Orthogonal Vectors, and More: Quickly Derandomizing
                  {R}azborov--{S}molensky},
  journal      = {{ACM} Trans. Algorithms},
  volume       = {17},
  number       = {1},
  pages        = {2:1--2:14},
  year         = {2021},
  _url          = {https://doi.org/10.1145/3402926},
  doi          = {10.1145/3402926},
  timestamp    = {Sat, 08 Jan 2022 02:22:07 +0100},
  biburl       = {https://dblp.org/rec/journals/talg/ChanW21.bib},
  bibsource    = {dblp computer science bibliography, https://dblp.org}
}

@article{AlonN96,
  author       = {Noga Alon and
                  Moni Naor},
  title        = {Derandomization, Witnesses for {B}oolean Matrix Multiplication and Construction
                  of Perfect Hash Functions},
  journal      = {Algorithmica},
  volume       = {16},
  number       = {4/5},
  pages        = {434--449},
  year         = {1996},
  _url          = {https://doi.org/10.1007/BF01940874},
  doi          = {10.1007/BF01940874},
  timestamp    = {Wed, 17 May 2017 14:25:12 +0200},
  biburl       = {https://dblp.org/rec/journals/algorithmica/AlonN96.bib},
  bibsource    = {dblp computer science bibliography, https://dblp.org}
}

@book{ChazelleBOOK,
  author       = {Bernard Chazelle},
  title        = {The Discrepancy Method: Randomness and Complexity},
  publisher    = {Cambridge University Press},
  year         = {2001},
  isbn         = {978-0-521-00357-5},
  bibsource    = {dblp computer science bibliography, https://dblp.org},
url = {https://www.cs.princeton.edu/~chazelle/book.html}
}

@incollection{MatousekSURV,
  author       = {Jir{\'{\i}} Matou{\v s}ek},
  editor       = {J{\"{o}}rg{-}R{\"{u}}diger Sack and
                  Jorge Urrutia},
  title        = {Derandomization in Computational Geometry},
  booktitle    = {Handbook of Computational Geometry},
  pages        = {559--595},
  publisher    = {North Holland / Elsevier},
  year         = {2000},
  _url          = {https://doi.org/10.1016/b978-044482537-7/50014-0},
  doi          = {10.1016/B978-044482537-7/50014-0},
  timestamp    = {Fri, 28 Jun 2019 09:10:30 +0200},
  biburl       = {https://dblp.org/rec/books/el/00/000100.bib},
  bibsource    = {dblp computer science bibliography, https://dblp.org}
}

@inproceedings{ColeH02,
  author       = {Richard Cole and
                  Ramesh Hariharan},
  _editor       = {John H. Reif},
  title        = {Verifying candidate matches in sparse and wildcard matching},
  booktitle    = {Proc. 34th Annual {ACM} Symposium on Theory of Computing (STOC)},
  pages        = {592--601},
  _publisher    = {{ACM}},
  year         = {2002},
  _url          = {https://doi.org/10.1145/509907.509992},
  doi          = {10.1145/509907.509992},
  timestamp    = {Tue, 06 Nov 2018 11:07:04 +0100},
  biburl       = {https://dblp.org/rec/conf/stoc/ColeH02.bib},
  bibsource    = {dblp computer science bibliography, https://dblp.org}
}

@book{BellmanBOOK,
author = {Richard E. Bellman},
title = {Dynamic Programming},
publisher = {Princeton University Press},
year = {1957}
}

@article{HorowitzS74,
  author       = {Ellis Horowitz and
                  Sartaj Sahni},
  title        = {Computing Partitions with Applications to the Knapsack Problem},
  journal      = {J. {ACM}},
  volume       = {21},
  number       = {2},
  pages        = {277--292},
  year         = {1974},
  _url          = {https://doi.org/10.1145/321812.321823},
  doi          = {10.1145/321812.321823},
  timestamp    = {Wed, 14 Nov 2018 10:35:25 +0100},
  biburl       = {https://dblp.org/rec/journals/jacm/HorowitzS74.bib},
  bibsource    = {dblp computer science bibliography, https://dblp.org}
}

@inproceedings{ChanL15,
  author       = {Timothy M. Chan and
                  Moshe Lewenstein},
  _editor       = {Rocco A. Servedio and
                  Ronitt Rubinfeld},
  title        = {Clustered Integer {3SUM} via Additive Combinatorics},
  booktitle    = {Proc. 47th Annual {ACM} on Symposium on Theory
                  of Computing (STOC)},
  pages        = {31--40},
  _publisher    = {{ACM}},
  year         = {2015},
  _url          = {https://doi.org/10.1145/2746539.2746568},
  doi          = {10.1145/2746539.2746568},
  timestamp    = {Tue, 06 Nov 2018 11:07:04 +0100},
  biburl       = {https://dblp.org/rec/conf/stoc/ChanL15.bib},
  bibsource    = {dblp computer science bibliography, https://dblp.org}
}

@article{Chan18,
  author       = {Timothy M. Chan},
  title        = {Improved Deterministic Algorithms for Linear Programming in Low Dimensions},
  journal      = {{ACM} Trans. Algorithms},
  volume       = {14},
  number       = {3},
  pages        = {30:1--30:10},
  year         = {2018},
  _url          = {https://doi.org/10.1145/3155312},
  doi          = {10.1145/3155312},
  timestamp    = {Sun, 19 Jan 2025 15:02:13 +0100},
  biburl       = {https://dblp.org/rec/journals/talg/CHAN18.bib},
  bibsource    = {dblp computer science bibliography, https://dblp.org}
}

@article{GyarmatiMR10,
  author       = {Katalin Gyarmati and
                  M{\'{a}}t{\'{e}} Matolcsi and
                  Imre Z. Ruzsa},
  title        = {A superadditivity and submultiplicativity property for cardinalities
                  of sumsets},
  journal      = {Comb.},
  volume       = {30},
  number       = {2},
  pages        = {163--174},
  year         = {2010},
  _url          = {https://doi.org/10.1007/s00493-010-2413-6},
  doi          = {10.1007/S00493-010-2413-6},
  timestamp    = {Wed, 22 Jul 2020 22:02:52 +0200},
  biburl       = {https://dblp.org/rec/journals/combinatorica/GyarmatiMR10.bib},
  bibsource    = {dblp computer science bibliography, https://dblp.org}
}
}

\appendix

\mysection{Reducing 0-1 Knapsack to Min-Plus Convolution}\label{sec:knapsack}

\newcommand{\MAX}{\mbox{\sc p-max}}
\newcommand{\R}{\mathbb{R}}

In the \emph{all-capacities} version of the \emph{0-1 knapsack problem},
we are given $n$ items where the $j$-th item has positive integer weight $w_j$ and positive real profit $p_j$, and we are given a number $t$.  For each $i\in [t]$, we want to compute $f[i]=\max\left\{\sum_{j\in J} p_j: \sum_{j\in J} w_j = i,\ 
J\subseteq\{1,\ldots,n\}\right\}$.

0-1 knapsack may be viewed as a two-dimensional variant of subset sum.
More precisely, we map each item to a point $(w_j,p_j)$ in two dimensions.
For a set of two-dimensional points $A\subseteq [t]\times\R^+$, define the operator $\MAX(A)=\left\{(w,p^*)\in A: \displaystyle p^*=\max_{(w,p)\in A}p\right\}$.
The problem is then equivalent to the following: given a set of $n$ two-dimensional points $X\subset [t]\times\R^+$,
compute $\MAX(\SUMS(X)\cap ([t]\times\R^+))$,
where $\SUMS(\cdot)$ is defined exactly as before (additions are now performed in two dimensions).
We may assume that $n\le t$ after an initial $O(n)$-time preprocessing (replacing $X$ with $\MAX(X\cap ([t]\times\R^+))$).

Cygan, Mucha, W\k{e}grzycki, and W\l{}odarczyk~\cite{CyganMWW19} and K\"unnemann, Paturi, and Schneider~\cite{KunnemannPS17} gave deterministic reductions from min-plus convolution to all-targets 0-1 knapsack.
In the reverse direction, Cygan et al.~\cite{CyganMWW19}
gave a \emph{randomized} reduction from all-targets 0-1 knapsack to min-plus convolution, 
by modifying Bringmann's randomized algorithm for subset sum~\cite{Bringmann17}.
We can obtain a deterministic reduction from all-targets 0-1 knapsack to min-plus convolution by modifying the algorithm in Section~\ref{sec:linear}.

The modification is mostly straightforward.
We mention some of the main changes:
\begin{itemize}
\item
Expressions such as $X\cap I$ should be reinterpreted as $X\cap (I\times\R^+)$.
\item
Most occurrences of $\SUMS_i(X)$ should be changed to $\MAX(\SUMS_i(X))$.
\item
Each sumset $A+B$ should be changed to $\MAX(A+B)$.
Given $A,B\subseteq [t]\times\R^+$, the computation of $\MAX(A+B)$ is equivalent to
max-plus convolution over two real vectors of length $O(t)$, which is equivalent to min-plus convolution by negation.
\item
In Definition~\ref{def:halver} on halvers, the property is now:
\begin{quote}
For every subset $S\subseteq X$ with $|S|\le k$ and $\SUM(S)\in \MAX(\SUMS_{\le k}(X))$, 
there exists a subset $\hatS\subseteq X$ with $|\hatS|\le k$, $\SUM(\hatS)=\SUM(S)$,
and $|\hatS\cap X'|, |\hatS\cap X''|\le |\hatS|/2 + \Delta$.
\end{quote}
\item
In Lemma~\ref{lem:collect} on canonical subset generation, the property is now:
\begin{quote}
For every subset $S\subseteq X$ with $\SUM(S)\in \MAX(\SUMS_{|S|}(X))$,
there exists a subset $\hatS\subseteq X$ such that $|\hatS|=|S|$, $\SUM(\hatS)=\SUM(S)$, and 
$\hatS$ is expressible as a union of $O((|S|/b + 1)\log u)$ disjoint canonical subsets in $\SSS$.
\end{quote}
Note that if $\SUM(S)\in \MAX(\SUMS_{|S|}(X))$, then for any interval $I$, we must also have
$\SUM(S\cap (I\times\R^+))\in \MAX_{|S\cap (I\times\R^+)|} (X\cap (I\times\R^+))$.
\item
In Lemma~\ref{lem:DC}, we are now computing a set $\ANS\subseteq [ku]\times\R^+$ such that
$\MAX(\SUMS_{\le k}(X))\subseteq\{(w,p'): p'\ge p,\ (w,p)\in\ANS\}$
and $\ANS\subseteq\SUMS(X)$.
\end{itemize}
The rest of the changes are easy to work out.

\begin{theorem}
If min-plus convolution for two real vectors of length $n$ could be computed in $T(n)$ deterministic time,
then the all-capacities version of the 0-1 knapsack problem could be solved in $\OO(T(t))$ deterministic time.
\end{theorem}

\end{document}